\documentclass[conference]{IEEEtran}
\usepackage[left=0.7in,right=0.7in,top=0.7in,bottom=0.4in]{geometry}
\usepackage{bm}
\usepackage{amsmath}
\usepackage{booktabs}
\usepackage{amssymb}
\usepackage{amsthm}
\usepackage{graphicx}
\usepackage{array}
\usepackage{nameref}
\usepackage{lipsum}

\usepackage{multirow}
\newtheorem{theorem}{Theorem}

\newtheorem{lemma}[theorem]{Lemma}
\newtheorem{auxiliary code}{Auxiliary Code}

\usepackage{subcaption}

\usepackage{authblk}
\usepackage[T1]{fontenc}
\usepackage{stfloats}
\usepackage{url}

\usepackage{color}
\usepackage{nccmath}
\usepackage[ruled]{algorithm2e}

\begin{document}
\title{Write and Read Channel Models for 1S1R Crossbar Resistive Memory with High Line Resistance}

\author{Zehui Chen and Lara Dolecek\\
	Department of Electrical and Computer Engineering - University of California, Los Angeles\\\textit {chen1046@ucla.edu, dolecek@ee.ucla.edu}\vspace{-1em}}
\date{}

\maketitle

\begin{abstract}
	Crossbar resistive memory with 1 Selector 1 Resistor (1S1R) structure is attractive for low-cost and high-density nonvolatile memory applications. As technology scales down to the single-nm regime, the increasing resistivity of wordline/bitline becomes a limiting factor to device reliability. This paper presents write/read communication channels while considering the line resistance and device variabilities by statistically relating the degraded write/read margins and the channel parameters. Binary asymmetric channel (BAC) models are proposed for the write/read operations. Simulations based on these models suggest that the bit-error rate of devices are highly non-uniform across the memory array. These models provide quantitative tools for evaluating the trade-offs between memory reliability and design parameters, such as array size, technology nodes, and aspect ratio, and also for designing coding-theoretic solutions that would be most effective for crossbar memory. 
\end{abstract}

\section{Introduction}
The crossbar resistive memory, whereby bistable memristors are placed at the crosspoint of wordlines and bitlines, is one promising candidate for the next generation nonvolatile memory due to its inherent $4F^2$ device density and its simple crossbar structure \cite{ielmini2015resistive}. Meanwhile, as technology scales down to single-digit-nm, simultaneously scaled wordline/bitline resistances increasingly become a limiting factor to device reliability and hence memory scalability \cite{liang2013effect}. 

Previous literature has extensively shown that even moderate line resistance significantly degrades the reliability of the write and read operations. The degradation of the write/read margins due to high line resistance for the wort-case memory cell, i.e., the cell that is furthest from the source and ground, are studied in \cite{liang2013effect,chen2016design,kim2015numerical}. The adverse effect of the line resistance on the write/read margins for cells across the memory array are studied in \cite{chen2013comprehensive,shin2010data} by solving a system of Kirchhoff's current law (KCL) equations. While these studies focused on the degradation of the write/read margins, it remains unclear how the degraded write/read margins affect the system level reliability metric, e.g., the bit-error rate (BER). In other words, channel models are not yet well-established for this problem.

It is demonstrated in \cite{chen2013comprehensive} that, when considering the line resistance in resistive memory, the write margins are nonuniform across the array, which leads to nonuniform reliability levels in the memory array. Designing ECCs for the worst-case often leads to overly conservative code design and is therefore not rate efficient. For example, in \cite{zorgui2019polar}, the authors designed a non-stationary polar code targeting channels with different reliability levels, which are characterized empirically by simulations. Moreover, \cite{zorgui2019polar} also showed that using more precise channel modeling, i.e., using the binary asymmetric channel (BAC) instead of the binary symmetric channel (BSC), provides an order of magnitude improvement in BER, which proves the necessity of precise channel models. 

In this work, we propose BAC models for writing to and reading from memory devices in crossbar memory, parameterized by device parameters, array size, wordline/bitline resistances and device location by statistically relating the degraded write/read margins of cells at different locations to the channel parameters. Our analytical channel models, which take into account the device location, provide quantitative tools for analyzing the aforementioned non-uniformity and the trade-off between device parameters and memory reliability. These models are therefore beneficial for system engineers when designing the next generation storage systems. For the read channel, we also propose an efficient procedure to compute an optimal read threshold. We showed that the row bit-error rate of the read channel is reduced by a large factor using the optimal read threshold. 

Previous studies on the write/read margins assume deterministic High Resistance State (HRS) and Low Resistance State (LRS) for the memory device whereas the HRS and LRS are nondeterministic in nature \cite{ji2015line,chen2011variability}. Our write/read channel models, which are derived probabilistically, allow us to take the resistance variability of LRS and HRS into consideration for more precise modeling. 

The content of this paper is organized as follows. Section II provides background on crossbar resistive memory and the write/read operation. The circuit models and the variabilities are also discussed in Section II. Section III presents the channel characterization for the write and read operations. Simulation results on the proposed channel models with various parameters are also presented in Section III. Section IV presents our study on the optimal read threshold and its simulation results. We conclude and discuss future research in Section V.

\section{Preliminary}
\subsection{1S1R Crossbar resistive memory Background and Model}
In crossbar resistive memory array, the logical state $0$ or $1$ is represented by the HRS or LRS of a memory cell, respectively. For a bipolar memristor, the state of a cell is switched from LRS to HRS (Reset Operation) or from HRS to LRS (Set Operation) by applying a positive or negative voltage across the memory cell, respectively. For the write operation, we consider the so called ``V/2'' write scheme (cf. \cite{chen2016design}) as it is usually more energy-efficient than the so called ``V/3'' write scheme. In particular, when writing to a selected cell, the wordline and bitline of the selected cell are biased at the write voltage ($V_{w\_set}$ or $V_{w\_reset}$) and 0, respectively, while other wordlines and bitlines are biased at half of the write voltage to prevent unintentional write, as shown in Fig. \ref{fig:write model}. 

\begin{figure}[h]
	\centering
	\begin{subfigure}{0.4\textwidth}
	\centering
	\includegraphics[scale=0.23]{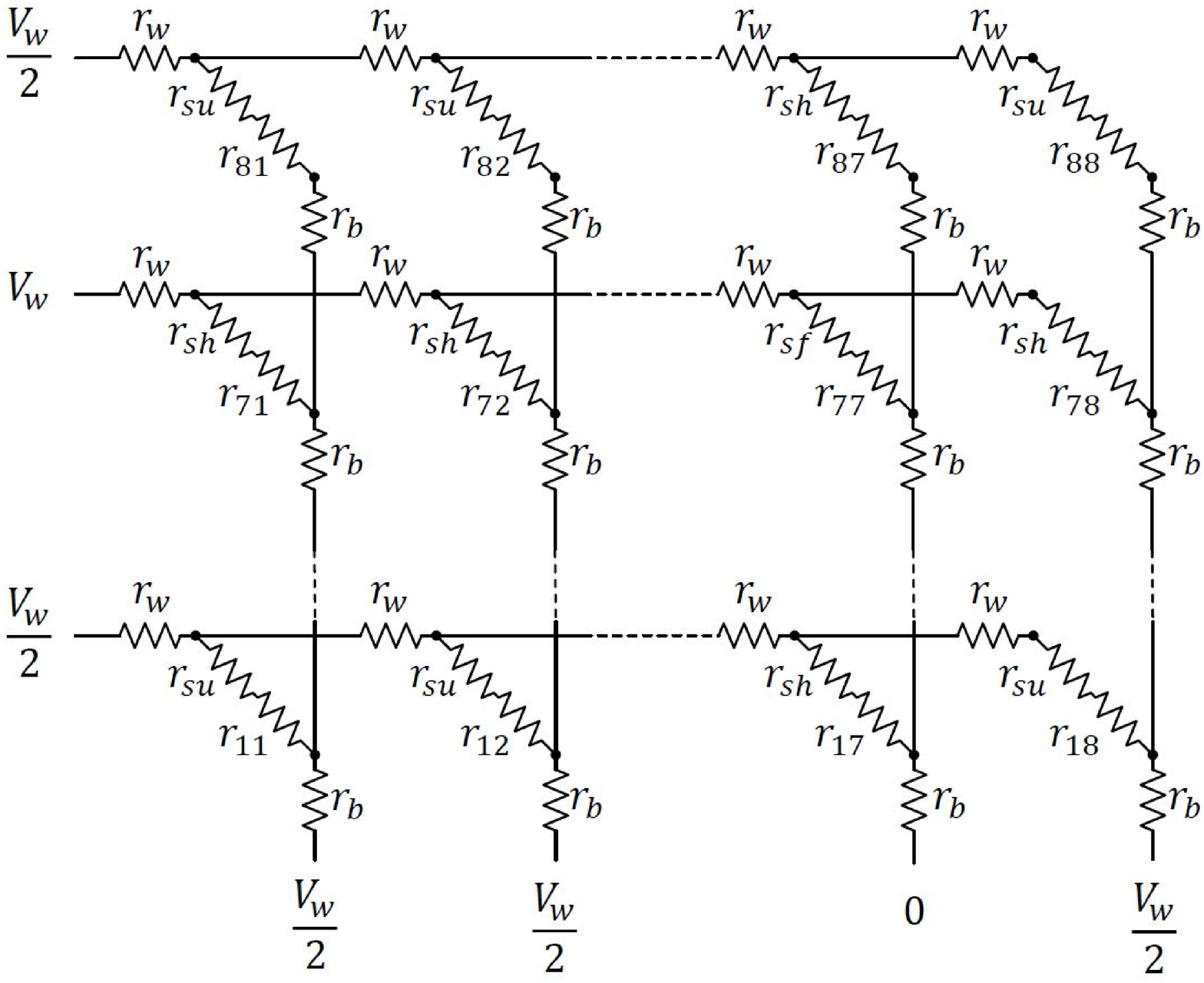}	
	\caption{Circuit model for writing to a $8\times8$ array.}
	\label{fig:write model}
	\end{subfigure}
	\begin{subfigure}{0.4\textwidth}
	\centering
	\includegraphics[scale=0.23]{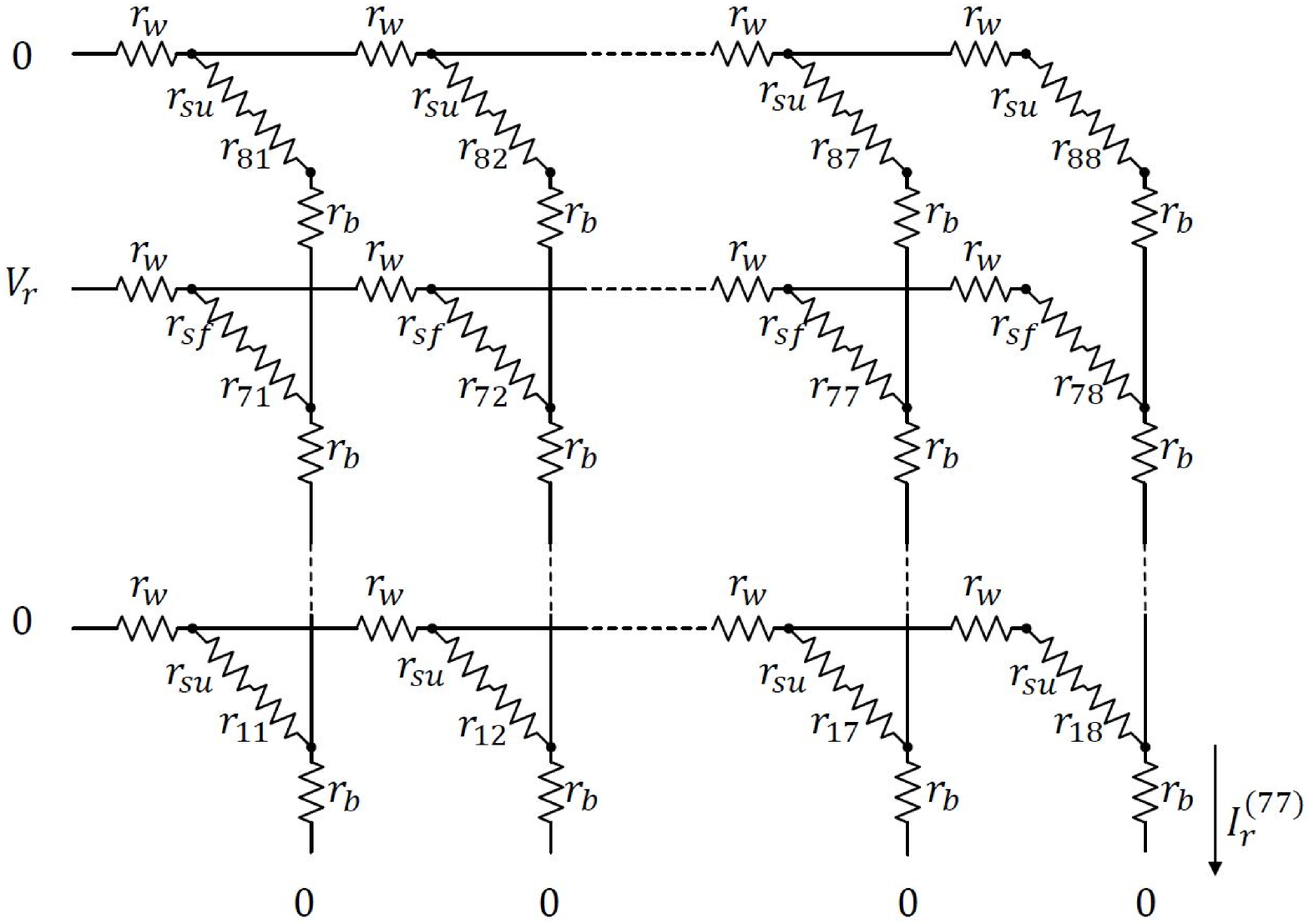}	
	\caption{Circuit model for reading from a $8\times8$ array.}
	\label{fig:read model}
	\end{subfigure}
	\caption{Example Circuit Models ($V_w$ denotes $V_{w\_set}$ or $V_{w\_reset}$)}
	\vspace{-1em}
\end{figure}

For the read operation, we consider the current-mode sensing scheme as it has a smaller latency compared with the voltage-mode sensing scheme \cite{chen2016design}. When reading a selected cell, a read voltage ($V_r$) is applied on its wordline and all other wordlines and bitlines are grounded. A current is sensed by the sensing amplifier located at the end of its bitline, and is used to determine the state of the selected cell, as shown in Fig. \ref{fig:read model}.

In this paper, we focus on crossbar resistive memory with the widely used 1 selector 1 resistor (1S1R) structure, where highly nonlinear selectors are connected in series with the memristors to prevent write and read disturbs. For both write and read operations, when the voltage across a selector is close to the applied voltage, we say that this selector is fully selected and we assume it has resistance $r_{sf}$; when the voltage across a selector is close to 0, we say that this selector is un-selected and we assume it has resistance $r_{su}$. For the write operation, since other cells on the wordline and bitline of the selected cell have voltage close to half of the write voltage across them, we say that the selectors for those cells are half-selected and we assume they have resistance $r_{sh}$. In general, $r_{sf}<<r_{sh}<r_{rs}$. An ideal selector have  parameters $r_{sf} = 0$ and $r_{sh}=r_{su}=\infty$. Our proposed model is a general one that does not have the ideal selector assumption. Meanwhile, since the main focus is the adverse effect of line resistance, we use the ideal selector assumption to provide mathematical insights in III.B and to simplify our simulations in III.D and IV.C. Throughout this work, we assume that the interconnect resistances of wordlines and bitlines are constant across the array, and they are denoted by $r_w$ and $r_b$ respectively.
\subsection{Memristor Variabilities and Models}
In this paper, we consider two variabilities of memristor, the non-deterministic write operation and the non-deterministic resistance value for each resistance state.  It is widely observed that the switching operations of memristor are stochastic and follow log-normal switching time distributions, with distribution parameters depend on the applied voltage \cite{medeiros2011lognormal,niu2012low}. Our models for the switching time distributions are adopted from \cite{medeiros2011lognormal} and more details are provided in Section III. 

Previous works (cf. \cite{liang2013effect} - \cite{shin2010data}) on the degradation of write and read margins due to high line resistance assume deterministic resistance states, e.g., HRS resistance is $10000\Omega$ and LRS resistance is $100\Omega$. Meanwhile, due to both device-to-device variation and cycle-to-cycle variation, the resistance of each state is highly non-deterministic \cite{ji2015line,chen2011variability}. To incorporate this variability into our reliability analysis, we use random variables to represent the resistance of the memory cells. Based on observations in \cite{ji2015line,chen2011variability}, we assume they are i.i.d. and their conditional distributions, conditioned on their states, follow log-normal distributions.  For example, let i.i.d. Bernoulli($q$) random variable $S_{ij}$ denote the state of cell $(i,j)$, with $S_{ij} = 1$ for LRS and $S_{ij} = 0$ for HRS. Let $R_{ij}$ be the associated random variable denoting the resistance of cell $(i,j)$. Then our model assumes:
\[\ln(R_{ij}|S_{ij}=1)\sim \mathcal{N}(\mu_{L},\sigma_{L}^2),\]
and
\[\ln(R_{ij}|S_{ij}=0)\sim \mathcal{N}(\mu_{H},\sigma_{H}^2).\]

\section{Channel Models}
Our proposed channel models are depicted in Fig. \ref{fig:cascaded_channel}. We model the write, read and cascaded channels as binary asymmetric channels (BACs). We discuss in the following subsections how the channel parameters are related to device parameters, line resistance and device location. 
\begin{figure}[h]
	\centering
	\includegraphics[scale=0.24]{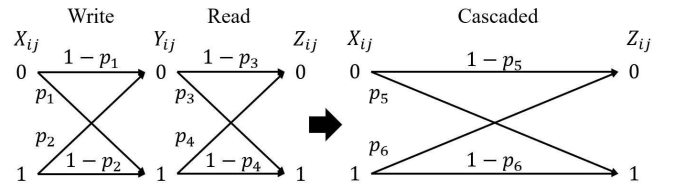}	
	\caption{Channel Models}
	\label{fig:cascaded_channel}
	\vspace{-1em}
\end{figure}

We denote the state we want to write to cell $(i,j)$ by $X_{ij}\in\{0,1\}$, the state actually written by $Y_{ij}\in\{0,1\}$ and the detected (for read operation) state by $Z_{ij}\in\{0,1\}$. Note that even though the information is stored as the resistance of a cell, we choose to use binarized state variable $Y_{ij}$ because firstly it {enables} us to utilize a well-known result in the literature \cite{medeiros2011lognormal} that characterizes the switching of a device; and secondly it allows mathematical tractability and the separation of the write/read channels. Information about the resistance of a cell is instead embedded in the resistance distribution. Also note that with the binarized state for a cell, multiple write/read operations are also independent. Therefore, the cascaded channel model is still valid if one writes to and reads from a cell multiple times.
\subsection{Write Channel}
In this section, we derive the write channel. We note that the write operation is affected by the previous state of cell $(i,j)$. We let this be $S^*_{ij}$ and the associated resistance value be $R^*_{ij}$. We assume that when the previous state is the same as the state we want to write, the write operation is always successful, i.e., 
$P(Y_{ij}=1|X_{ij}=1,S^*_{ij}=1)=1,$
and $P(Y_{ij}=0|X_{ij}=0,S^*_{ij}=0)=1.$

When the previous state is not the same as the state we want to write, a sufficient write voltage and a sufficient write time {are} required to change the state of the cell. Due to high line resistances, the effective write voltage on a cell could be much smaller than the desired write voltage, i.e., the write margin is decreased. We denote the effective write voltage on a cell $(i,j)$ as $\tilde{V}_w(r^*_{ij},i,j)$ where $r^*_{ij}$ is a realization of $R^*_{ij}$. With a method similar to the one described in \cite{chen2013comprehensive}, $\tilde{V}_w(r^*_{ij},i,j)$ can be obtained by solving a system of KCL equations using the circuit model described in Section II.A. We map the degraded write margin to the decreased write reliability by considering the log-normal switching time distribution, adopted from \cite{medeiros2011lognormal}. With fixed switching times $t_{set}$ and $t_{reset}$, the log-normal switching time distributions lead to the following:
\begin{equation}
\label{equ:set_success}
\begin{split}
&P(Y_{ij}=1|X_{ij}=1,S^*_{ij}=0,R^*_{ij}=r^*_{ij})\\
=&1-Q\left(\frac{\ln t_{set} -\ln(\tau^{(ij)}_{set})}{\sigma_{set}}\right),
\end{split}
\end{equation} 
and
\begin{equation}
\begin{split}
\label{equ:reset_success}
&P(Y_{ij}=0|X_{ij}=0,S^*_{ij}=1,R^*_{ij}=r^*_{ij})\\
=&1-Q\left(\frac{\ln t_{reset} -\ln(\tau^{(ij)}_{reset})}{\sigma_{reset}}\right),
\end{split}
\end{equation} 
where $Q(\cdot)$ is the $Q$-function, i.e., $Q(x) = \frac{1}{\sqrt{2\pi}}\int_{x}^{\infty}\exp(-\frac{u^2}{2})du$. Parameters $\sigma^2_{set}$ and $\sigma^2_{reset}$ are the variance of the normal distributions associated with the set and reset switching time distribution, respectively, which are independent of $\tilde{V}_w(r^*_{ij},i,j)$, according to \cite{medeiros2011lognormal}. Parameters $\tau^{(ij)}_{set}$ and $\tau^{(ij)}_{reset}$ are the median of the set and reset switching time, respectively. Note that in the above equations, to be consistent with the existing literature \cite{medeiros2011lognormal,niu2012low}, we use the median parameterization of the log-normal distribution. According to the literature, the medians of the switching times ($\tau^{(ij)}_{set}$ and $\tau^{(ij)}_{reset}$ in $\mu s$) are exponentially dependent on the effective write voltage. We therefore parameterize the medians as following:
\[\ln\left(\tau^{(ij)}_{set}\right) = \alpha_{set}\tilde{V}_w(r^*_{ij},i,j)+\beta_{set},\]
and
\[\ln\left(\tau^{(ij)}_{reset}\right) = \alpha_{reset}\tilde{V}_w(r^*_{ij},i,j)+\beta_{reset}.\]

Using (\ref{equ:set_success}), (\ref{equ:reset_success}) and marginalizing over the conditionally log-normally distributed random variable $R^*_{ij}$, we get:
\begin{equation}
\label{equ:set_fail}
\begin{split}
&P(Y_{ij}=0|X_{ij}=1,S^*_{ij}=0)=\int_{-\infty}^{\infty}\frac{1}{\sqrt{2\pi}r^*_{ij}\sigma_{H}}\\
&\times \exp\left[-\frac{\left(\ln r^*_{ij}-\mu_H\right)^2}{2\sigma_{H}^2}\right]Q\left(\frac{\ln t_{set} -\ln(\tau^{(ij)}_{set})}{\sigma_{set}}\right)dr^*_{ij},
\end{split}
\end{equation} 
and
\begin{equation}
\label{equ:reset_fail}
\begin{split}
&P(Y_{ij}=1|X_{ij}=0,S^*_{ij}=1)=\int_{-\infty}^{\infty}\frac{1}{\sqrt{2\pi}r^*_{ij}\sigma_{L}}\\
&\times \exp\left[-\frac{\left(\ln r^*_{ij}-\mu_L\right)^2}{2\sigma_{L}^2}\right]Q\left(\frac{\ln t_{reset} -\ln(\tau^{(ij)}_{reset})}{\sigma_{reset}}\right)dr^*_{ij}.
\end{split}
\end{equation} 

Putting (\ref{equ:set_fail}) and (\ref{equ:reset_fail}) together with the prior symbol probability $q=P(S^*_{ij}=0)$, we arrive at the write binary asymmetric channel, depicted in Fig. \ref{fig:cascaded_channel}, for the write operation with the following channel parameters:
\begin{equation}
\label{equ:p1}
p^{(ij)}_1 = (1-q)P(Y_{ij}=1|X_{ij}=0,S^*_{ij}=1),
\end{equation}
and
\begin{equation}
\label{equ:p2}
p^{(ij)}_2 = qP(Y_{ij}=0|X_{ij}=1,S^*_{ij}=0).
\end{equation}
Here and elsewhere, we use superscript $(ij)$ to highlight that the channel parameters are dependent on the cell location $(i,j)$. 

Through equations (\ref{equ:set_success}) - (\ref{equ:p2}), we are able to relate the write margin $\tilde{V}_w(r^*_{ij},i,j)$ to the BER of the write channel. For example, comparing the best-case cell to the worst-case cell in the example in Section III.D Fig. \ref{fig:heat_map}, we observe that the write margin for Reset is dropped from $4.9 V$ to $1.64 V$ while the write BER is increased from $3.35\times10^{-4}$ to $1.75\times10^{-2}$, thus providing further evidence that location dependent BER analysis matters.

\subsection{Read Channel}
When reading from the cell $(i.j)$, we consider the current-mode sensing scheme and a fixed threshold detector. Let $I^{(ij)}_r$ be the current sensed by the sensing amplifier, which can be also calculated by solving a system of KCL equations. $I^{(ij)}_r$ is hence dependent on the cell location, the resistance of the selected cell, and the resistances of unselected cells. Let $I_{th}$ be the threshold current. The threshold detector is as follows:
\begin{equation}
Z_{ij}=\begin{cases}
0, I^{(ij)}_r\leq I_{th},\\
1, I^{(ij)}_r>I_{th}.
\end{cases}
\end{equation}
With the threshold detector above, the decision error probabilities are:
\begin{equation}
\label{equ:0_read_error}
P(Z_{ij}=1|Y_{ij}=0) = P(I^{(ij)}_r> I_{th}|Y_{ij}=0);
\end{equation}

\begin{equation}
\label{equ:1_read_error}
P(Z_{ij}=0|Y_{ij}=1) = P(I^{(ij)}_r\leq I_{th}|Y_{ij}=1).
\end{equation}
This leads to the read binary asymmetric channel, depicted in Fig. \ref{fig:cascaded_channel}, for the read operation with $p^{(ij)}_3 = P(Z_{ij}=1|Y_{ij}=0)$ and $p^{(ij)}_4 = P(Z_{ij}=0|Y_{ij}=1)$.

\subsubsection{Closed form Expression with Ideal Selectors}
Since we need to solve a system of equations to get $I^{(ij)}_r$, equations (\ref{equ:0_read_error}) and (\ref{equ:1_read_error}) are not sufficient as they do not give closed-form expressions for the channel parameters. However, if we consider ideal selectors, closed-form expressions can be derived. { Note that for the analysis with an unideal selector, one can still use the general characterizations in equations (\ref{equ:0_read_error}) and (\ref{equ:1_read_error}) and find $I^{(ij)}_r$ by solving a system of KCL equations. Moreover, it is reasonable to assume ideal characteristics of the selector for the analysis of the line resistance in the 1S1R structure as near ideal selector properties are demonstrated by industry in  \cite{jo20143d}.}

With ideal selectors, the part of the circuit connected to the un-selected cells can be neglected, resulting in a simplified circuit with just the selected cell and its wordline/bitline. With this simplified circuit, $I^{(ij)}_r$ is a function of the random variable $R_{ij}$, which represents the resistance of the selected cell. We therefore have:
\begin{equation}
\label{equ:I_r}
I^{(ij)}_r = \frac{V_r}{ir_b+jr_w+R_{ij}}.
\end{equation}
Plugging (\ref{equ:I_r}) into (\ref{equ:0_read_error}) and (\ref{equ:1_read_error}), and using the assumption that $R_{ij}$ is conditionally (on $Y_{ij}$) log-normally distributed, we obtain the following closed form expression for $p_3$ and $p_4$:
\begin{equation}
\label{equ:p3}
\begin{split}
p^{(ij)}_3 &= P\left(\frac{V_r}{ir_b+jr_w+R_{ij}}>I_{th}|Y_{ij}=0\right)\\
&=P\left(R_{ij}<\frac{V_r}{I_{th}}-ir_b-jr_w|Y_{ij}=0\right)\\
&=Q\left(\frac{\mu_{H}-\ln\left(\frac{V_r}{I_{th}}-ir_b-jr_w\right)}{\sigma_{H}}\right),
\end{split}
\end{equation}
and similarly
\begin{equation}
\label{equ:p4}
\begin{split}
p^{(ij)}_4 =Q\left(\frac{\ln\left(\frac{V_r}{I_{th}}-ir_b-jr_w\right)-\mu_{L}}{\sigma_{L}}\right).
\end{split}
\end{equation}

Define $R_{th} = \frac{V_r}{I_{th}}$. From equations (\ref{equ:p3}) and (\ref{equ:p4}), we observe that $R_{th}$ is the effective decision threshold between the HRS and LRS distribution in the resistance domain, when there {is} no line resistance, i.e., $r_w=r_b=0$. We can therefore interpret the adverse effect of line resistances during the read operation as follows: the effective read threshold in resistance domain is shifted to the left by the total accumulated line resistance. This shift results in a higher bit-error rate if $R_{th}$ is set to be the optimal decision threshold without considering the line resistance. 

The read margin is defined by the difference between the sensed current of a HRS cell and the sensed current of a LRS cell. Using equations (\ref{equ:I_r}) - (\ref{equ:p4}), we can now relate the read margin to the read BER. For example, comparing the best-case cell to the worst-case cell in the example in Section III.D Fig. \ref{fig:heat_map}, we observe that the read margin is dropped from $296 \mu A$ to $95 \mu A$ while the write BER is increased from $4.29\times10^{-4}$ to $7.33\times10^{-2}$, again demonstrating the need of a location dependent model.

\subsection{Cascaded Channel and Channel Capacity}
Combining the results of the previous two subsections, we get a cascaded channel for a single memory cell. The cascaded channel is a binary asymmetric channel and it is depicted in Fig. \ref{fig:cascaded_channel}, with $p^{(ij)}_5 = p^{(ij)}_1(1-p^{(ij)}_4)+(1-p^{(ij)}_1)p^{(ij)}_3$ and $p^{(ij)}_6 = p^{(ij)}_2(1-p^{(ij)}_3)+(1-p^{(ij)}_2)p^{(ij)}_4$.

The capacity of this cascaded channel for cell $(i,j)$ is as follows:
\begin{equation}
\label{equ:cell_capacity}
\begin{split}
C_{ij} &= \max_{q}\,I(X_{ij};Z_{ij})\\
&=\max_{q}\,\Bigg[h\left(q\left(1-p^{(ij)}_5\right)+(1-q)p^{(ij)}_6\right)\\
&\quad\quad\quad-qh\left(p^{(ij)}_5\right)-(1-q)h\left(p^{(ij)}_6\right)\Bigg],
\end{split}
\end{equation}
where $h(\cdot)$ is the binary entropy function. Because $p^{(ij)}_1$ and $p^{(ij)}_2$ are dependent on $q$, the closed form capacity result for a standard BAC does not hold. The channel capacity therefore need to be evaluated with a numerical method, e.g., the Blahut-Arimoto algorithm, as further presented in Subsection D. 

\subsection{Simulations Results}
Based on our models presented in the previous subsections, we simulate multiple arrays to explore how memory parameters affect the memory reliability metrics, such as the bit-error rate (BER) and the averaged capacity. We calculate the averaged capacity by averaging the capacities of cells given by equation (\ref{equ:cell_capacity}); this result serves as an indicator of what fraction of the input data can be reliably stored in memory. Since this work is mainly focused on the adverse effect of the line resistance, we only vary the array size, aspect ratio, and line resistance in our simulations. Other memory parameters are kept the same and are summarized in Table \ref{table1}. As an illustrative example, the parameters are chosen to represent a moderate reliability level, with a BER on the order of $10^{-3}$ in the best case scenario. { The considered line resistances range from $10\Omega$ to $100\Omega$, in accordance with the interconnect resistance values of interest for moderate technology nodes \cite{liang2013effect}.  The chosen standard deviation (0.3) of LRS and HRS distribution is experimentally observed in \cite{ji2015line}. The chosen switching parameters are based on \cite{medeiros2011lognormal}; we use the same parameters for reset and set operations for simplified analysis.
}
\begin{table}[h!]
	\centering
	\renewcommand{\arraystretch}{1}
	\begin{tabular}{c|c|c}
		\hline
		\hline
		Symbol 		& Parameters 						& Values \\ 
		$m,n$			& Array Size ($m\times n$)          &   varies               \\ 
		$V_{w\_set}$			& Set voltage       &    -5V              \\ 
		$V_{w\_reset}$			& Reset voltage       &    5V              \\ 
		$V_{r}$			& Read voltage       &    3V              \\ 
		$q$		& Prior symbol probability of 0          & $0.5$                 \\ 
		$r_w$		&Wordline interconnect resistance       &   $10\Omega-100\Omega$               \\ 
		$r_b$		&Bitline interconnect resistance       &    $10\Omega-100\Omega$              \\ 
		$r_{sf}$		&Fully selected selector resistance    &         $0$       \\ 
		$r_{sh}$		&Half selected selector resistance    &         $\infty$       \\ 
		$r_{su}$		&Unselected selector resistance   &         $\infty$       \\ 
		$\mu_L$ 	& Associated mean of LRS distribution     	&    $4\ln(10)$              \\
		$\mu_H$ 	& Associated mean of HRS distribution      	&    $6\ln(10)$              \\
		$\sigma_L$ 	& Associated std of LRS distribution       	&      $0.3\ln(10)$            \\
		$\sigma_H$ 	& Associated std of HRS distribution       	&      $0.3\ln(10)$           \\
		$\alpha_{set}$ 	& Parameter for the median set time    	&      $0.25$            \\
		$\beta_{set}$ 	& Parameter for the median set time    	&       $4.25$           \\
		$\alpha_{reset}$ 	& Parameter for the median reset time    	&    $-0.25$              \\
		$\beta_{reset}$ 	& Parameter for the median reset time    	&     $4.25$             \\
		$\sigma_{set}$ 	& Associated std of set time distribution       	&        $0.5$          \\
		$\sigma_{reset}$ 	& Associated std of reset time distribution       	&        $0.5$          \\
		$t_{set}$ 	&  Switching time for set operation      	&    $100\mu s$              \\
		$t_{reset}$ 	&  Switching time for reset operation      	&      $100\mu s$            \\
		$I_{th}$ 	& Read decision threshold       	&      $30\mu A$            \\
		\hline
		\hline
	\end{tabular}
	\caption{Summary of Parameters.}
	\label{table1}
	\vspace{-1em}
\end{table}

\begin{figure}[h]
	\centering
	\begin{subfigure}{0.4\textwidth}
		\centering
		\includegraphics[scale=0.25]{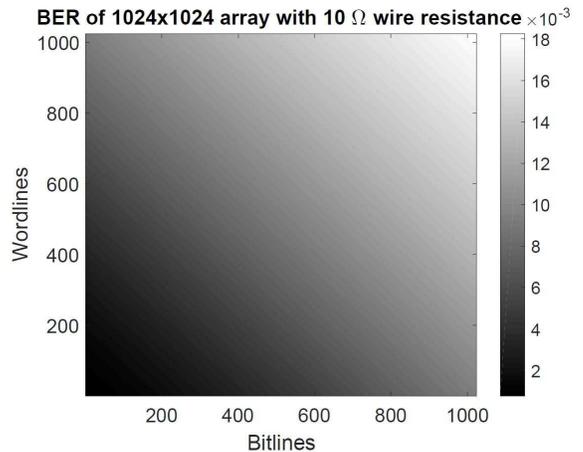}	
		\caption{Heatmap of BERs for a 1024x1024 array}
		\label{fig:heat_map}
	\end{subfigure}
	\begin{subfigure}{0.4\textwidth}
		\centering
		\includegraphics[scale=0.25]{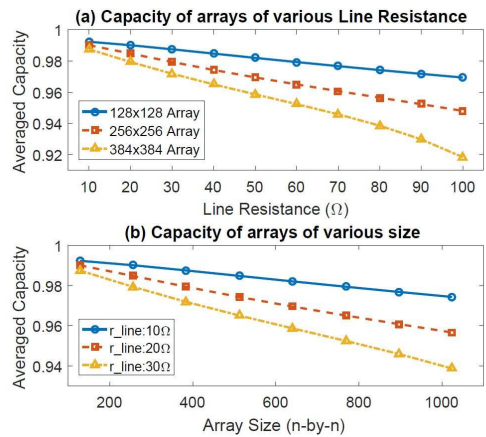}	
		\caption{Capacity results with various sizes and line resistances}
		\label{fig:Capacity_vs_N_and_vs_R}
	\end{subfigure}
	\caption{Simulation Results}
	\vspace{-1em}
\end{figure}
In Fig. \ref{fig:heat_map}, we first present the BER of each cell in a $1024\times1024$ array to illustrate the spatial variation of reliability due to the line resistance. According to \cite{liang2013effect}, the chosen $10\Omega$ line resistance corresponds to the resistance per junction of Cu wire with 20nm technology nodes. With this moderate line resistance, we observe an order of magnitude BER difference between the best-case cell, located closest to the voltage source, and the worst-case cell, located furthest from the voltage source. Due to line resistance, the cell which is further from the source and sensing amplifier, suffers from a lower voltage delivery during the write operation and a higher resistance interference during the read operation, thus has a larger BER. 

Next, in Fig. \ref{fig:Capacity_vs_N_and_vs_R}, we present the averaged capacity per cell for arrays with various size and line resistances, with aspect ratio fixed to be $1$. We observe that a larger line resistance, which corresponds to a smaller technology node, deteriorates the averaged capacity almost linearly. This trade-off thus must be taken into consideration when scaling the memory, as it is shown in \cite{liang2013effect} that the line resistance scales exponentially with respect to the technology node. Also note that when the accumulated line resistance of the worst-case cell gets close to the effective resistance threshold, i.e., when $nr_w+mr_b$ is close to $R_{th}$, the averaged capacity deteriorates faster, as from (\ref{equ:p4}), when $nr_w+mr_b>R_{th}$, reading from a LRS cell correctly is impossible. This explains the rapid dropping at the end of the curve in Fig. \ref{fig:Capacity_vs_N_and_vs_R} for the $384\times384$ array. From  Fig. \ref{fig:Capacity_vs_N_and_vs_R}, we notice that the averaged capacity also deteriorates almost linearly with respect to the array size. This effect is thus a limiting factor for the realization of a large memory array.

\begin{table}[h!]
	\renewcommand{\arraystretch}{1.1}
	\centering
	\begin{tabular}{|c|c|c|c|}
		\hline
		Array Size & $128\times128$ &$64\times512$  &$32\times512$  \\ \hline
		Averaged Capacity   & $0.9924$ & $0.9918$ & $0.9897$ \\ \hline\hline
		Array Size & $16\times1024$ &$8\times2048$  &$4\times4096$  \\ \hline
		Averaged Capacity   & $0.9845$ & $0.9745$ & $0.9573$ \\ \hline
	\end{tabular}
	\caption{Capacity of arrays with different aspect ratios.}
	\label{table2}
	\vspace{-1em}
\end{table}

We further investigate how the aspect ratio affects the averaged capacity by simulating arrays with the same number of cells but different aspect ratios. In Table \ref{table2}, with a total of $16384$ cells, the square array (aspect ratio = 1) has the largest averaged capacity and the $4\times 4096$ array, which has the largest aspect ratio, has the lowest averaged capacity. Intuitively, this can be explained by a larger possible cumulative line resistance $nr_w+mr_b$ in an array with a larger aspect ratio. This observation presents a trade-off between the sometimes desired high aspect ratio and a high averaged capacity. 

\section{Optimal Read Threshold}
In Section III.B, we observe that the channel parameters $p_3^{(ij)}$ are $p_4^{(ij)}$ are dependent on the read current threshold $I_{th}$ --- a user defined parameter that can be optimized for lower row bit-error rate. In this section, we study the optimal threshold for each cell and for an entire array with the goal of reducing the  row bit-error rate. We choose to optimize the read resistance threshold $R_{th}=\frac{V_r}{I_{th}}$ as it is equivalent to optimizing $I_{th}$ with fixed read voltage $V_r$. We define $R_{th0}$ be the optimal threshold when no line resistance is considered, i.e.,
\begin{equation}
\small
\label{equ:Rth0_opt}
R_{th0}=\underset{R_{th}}{\text{argmin}}\, qQ\left(\frac{\mu_H-\ln(R_{th})}{\sigma_H}\right)+(1-q)Q\left(\frac{\ln(R_{th})-\mu_L}{\sigma_L}\right)
\normalsize
\end{equation}
$R_{th0}$ can be calculated using standard result from estimation theory. $R_{th0}$ is clearly suboptimal when the line resistance is non-negligible. 

\subsection{Optimal Threshold for Each Cell}
One simple read scheme is to use the optimal resistance thresholds for cells at different locations. We call this the different threshold for each cell (DTEC) scheme. Define the optimal threshold for the cell $(i,j)$ to be $R^{ij}_{th}$. With the objective of minimizing $P(Z_{ij}~=Y_{ij})$, and using equation (\ref{equ:p3}) and equation (\ref{equ:p4}), we have:
\begin{equation}
\label{(equ:Rthij_opt)}
\begin{split}
R^{ij}_{th}&=\underset{R_{th}}{\text{argmin}}\, qQ\left(\frac{\mu_H-\ln(R_{th}-ir_b-jr_w)}{\sigma_H}\right)\\
&+(1-q)Q\left(\frac{\ln(R_{th}-ir_b-jr_w)-\mu_L}{\sigma_L}\right).
\end{split}
\end{equation}
Comparing equations (\ref{equ:Rth0_opt}) and (\ref{(equ:Rthij_opt)}), we get
\begin{equation}
\label{equ:opt_thr_cell}
R^{ij}_{th} = R_{th0}+ir_b+jr_w.
\end{equation}
This result is intuitive as we need to shift the threshold to the right in order to compensate for the adverse effect of the cumulative line resistance. 
\subsection{Optimal Threshold for An Array}
Requiring different thresholds for cells in a $m\times n$ array may not be desirable for circuit designers as doing this may require a lot more comparators. In typical memory design, cells on the same wordline share the same sensing amplifier, so one threshold for each column is a reasonable choice. To further simplify the memory design, one may even use the same threshold for the entire memory array. Therefore, it is of interest to find the optimal threshold that minimizes the averaged BER for an entire array or a sub-array (such as a column). In this subsection, we deal with the optimal threshold for an array first; this result readily generalizes to any sub-array. We call these schemes the same threshold for many cells (STMC) schemes.

With the objective of minimizing $\frac{1}{mn}\sum_{i,j}P(Z_{ij}~=Y_{ij})$, the optimal threshold for an array is defined as
\small
\begin{equation}
\label{(equ:Rtharray_opt)}
\begin{split}
R_{th\_array} = &\underset{R_{th}}{\text{argmin}}\frac{1}{mn}\sum_{i,j}^{m,n}\Bigg[ qQ\left(\frac{\mu_H-\ln(R_{th}-ir_b-jr_w)}{\sigma_H}\right)\\
&+(1-q)Q\left(\frac{\ln(R_{th}-ir_b-jr_w)-\mu_L}{\sigma_L}\right)\Bigg].
\end{split}
\end{equation}
\normalsize

{While it may be possible to heuristically optimize the single parameter $R_{th}$ in practice, we wish to provide an analytical solution based on certain approximations and an efficient iterative search algorithm. Analytically,} the optimization problem in (\ref{(equ:Rtharray_opt)}) is hard to solve for  as it involves a summation of Q-functions. We instead replace this objective function with its upper bound and try to minimize this upper bound. Using Jensen's inequality and the fact that the Q-function is concave for a positive argument, we have the following bound:
\begin{equation}
\begin{split}
&\frac{1}{mn}\sum_{i=1}^{m}\sum_{j=1}^{n}\Bigg[ qQ\left(\frac{\mu_H-\ln(R_{th}-ir_b-jr_w)}{\sigma_H}\right)\\
&+(1-q)Q\left(\frac{\ln(R_{th}-ir_b-jr_w)-\mu_L}{\sigma_L}\right)\Bigg]\\
\leq&qQ\left(\frac{\mu_H-A}{\sigma_H}\right)+(1-q)Q\left(\frac{A-\mu_L}{\sigma_L}\right),
\end{split}
\end{equation}
where
\[A = \frac{1}{mn}\sum_{i=1}^{m}\sum_{j=1}^{n}\Big[\ln(R_{th}-ir_b-jr_w)\Big].\]
The gap between the two sides of this inequality is small when the Q-functions are close to being linear, which is indeed the case for Q-functions with large arguments. 

We reformulate the problem using the above inequality:
\begin{equation}
\label{(equ:Rtharray_opt2)}
R_{th\_array}=\underset{R_{th}}{\text{argmin}}\,qQ\left(\frac{\mu_H-A}{\sigma_H}\right)+(1-q)Q\left(\frac{A-\mu_L}{\sigma_L}\right).
\end{equation}
Comparing equations (\ref{equ:Rth0_opt}) and (\ref{(equ:Rtharray_opt2)}), the problem becomes finding $R_{th\_array}$ such that 
\begin{equation}
\label{equ:log_equ}
\ln(R_{th0})=\frac{1}{mn}\sum_{i=1}^{m}\sum_{j=1}^{n}\Big[\ln(R_{th\_array}-ir_b-jr_w)\Big].
\end{equation}
Equation (\ref{equ:log_equ}) is hard to solve since it contains a summation of logarithm functions. We provide both an approximation to this equation that is easier to solve, as well as an iterative algorithm that produces a solution to the original equation.

Approximating each terms in the right hand side summation by $\ln(R_{th\_array}-\frac{m+1}{2}r_b-\frac{n+1}{2}r_w)$, the approximate solution of (\ref{equ:log_equ}) can be solved:
\begin{equation}
\label{equ:log_equ_sol_1}
R_{th\_array}\approx R_{th\_array\_appx}=R_{th0}+\frac{m+1}{2}r_b+\frac{n+1}{2}r_w.
\end{equation}
This approximation can be also interpreted as averaging the optimal thresholds, given by equation (\ref{equ:opt_thr_cell}), of all cells.

We propose Algorithm 1 to compute the exact solution of equation (\ref{equ:log_equ}).

\begin{algorithm}[h]
	\caption{STMC threshold solver algorithm}
	\SetAlgoLined
	1. Initialize $R^{(0)}_{th}=R_{th0},{l}=0$.\\
	2. ${l}={l}+1$.\\
	\quad Calculate $R^{({l})}_{th}$ such that $\ln(R^{({l})}_{th}) = \ln(R_{th0})-\frac{1}{mn}\sum_{i=1}^{m}\sum_{j=1}^{n}\ln\left(1-\frac{ir_b+jr_w}{R^{({l}-1)}_{th}}\right)$.\\
	3. Repeat step 2 until $k$ iterations is reached or $R^{({l})}_{th}-R^{({l}-1)}_{th}\leq\epsilon$. Let $R_{th\_array} = R^{({l})}_{th}$.
\end{algorithm}

The STMC threshold solver algorithm is inspired by how the summation of logarithm functions is handled in the Expectation Maximization algorithm. The convergence of theSTMC threshold solver algorithm is demonstrated the following lemma.
\begin{lemma}
	The STMC threshold solver algorithm converges to the solution of equation (\ref{equ:log_equ}). 
\end{lemma}
\begin{proof}
	Equation (\ref{equ:log_equ}) can be rewritten as:
	\[\ln(R_{th\_array}) = \ln(R_{th0})-\frac{1}{mn}\sum_{i=1}^{m}\sum_{j=1}^{n}\ln\left(1-\frac{ir_b+jr_w}{R_{th\_array}}\right).\]
	With $R^{(0)}_{th}\geq  R_{th\_array}$, the following inequality holds:
	\[
	\begin{split}
	\ln(R_{th\_array})&\geq \ln(R^{(1)}_{th})\\ &= \ln(R_{th0})-\frac{1}{mn}\sum_{i=1}^{m}\sum_{j=1}^{n}\ln\left(1-\frac{ir_b+jr_w}{R^{(0)}_{th}}\right).
	\end{split}\]
	By induction, $\ln(R_{th\_array})$ can be bounded:
	$\ln(R^{(2j+1)}_{th})\leq \ln(R_{th\_array})\leq \ln(R^{(2j)}_{th}), j\in\mathbb{Z}.$
	With {$R^{(2)}_{th}> R^{(0)}_{th}=R_{th0}$}, since $\ln(R^{({l})}_{th}) = \ln(R_{th0})-\frac{1}{mn}\sum_{i=1}^{m}\sum_{j=1}^{n}\ln\left(1-\frac{ir_b+jr_w}{R^{({l}-1)}_{th}}\right)$, the subsequent upper bound and lower bound are shrinking toward $\ln(R_{th\_array})$, i.e., $\ln(R^{(2j+2)}_{th})< \ln(R^{(2j)}_{th})$ and $\ln(R^{(2j+3)}_{th})>\ln(R^{(2j+1)}_{th})$. Therefore, the STMC threshold solver algorithm converges.
\end{proof}
{For the simulations we performed in the next section,  the STMC threshold solver algorithm converges in $5$ iterations which is more effective than performing a line search to solve equation (\ref{(equ:Rtharray_opt)}) empirically.}

\subsection{Simulation Results}
We simulate arrays to examine how the averaged read BERs are affected by different read thresholding schemes. Four different thresholding schemes are compared: the naive scheme where $R_{th0}$ is used as the threshold; the DTEC scheme; the STMC scheme with the approximated solution in (\ref{equ:log_equ_sol_1}); and the STMC scheme with the exact solution solved by Algorithm 1. Unless otherwise mentioned, we use parameters from Table \ref{table1}. We first vary the array size and fix $r_b=r_w=30\Omega$ and report the result in Fig. \ref{fig:thr_BER_vs_n}. We also vary the wire resistance and fix $n=m=1024$ and report the result in Fig. \ref{fig:thr_BER_vs_r_wire}. 

\begin{figure}[h]
	\centering
	\begin{subfigure}{0.4\textwidth}
	\centering
	\includegraphics[scale=0.25,clip]{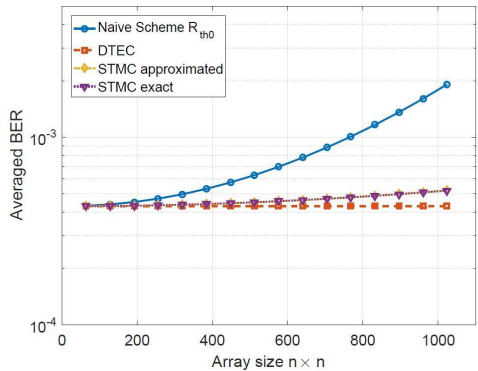}
	\caption{Averaged read BER v.s. array size using different thresholding scheme.}
	\label{fig:thr_BER_vs_n}
	\end{subfigure}
	\begin{subfigure}{0.4\textwidth}
	\centering
	\includegraphics[scale=0.25,clip]{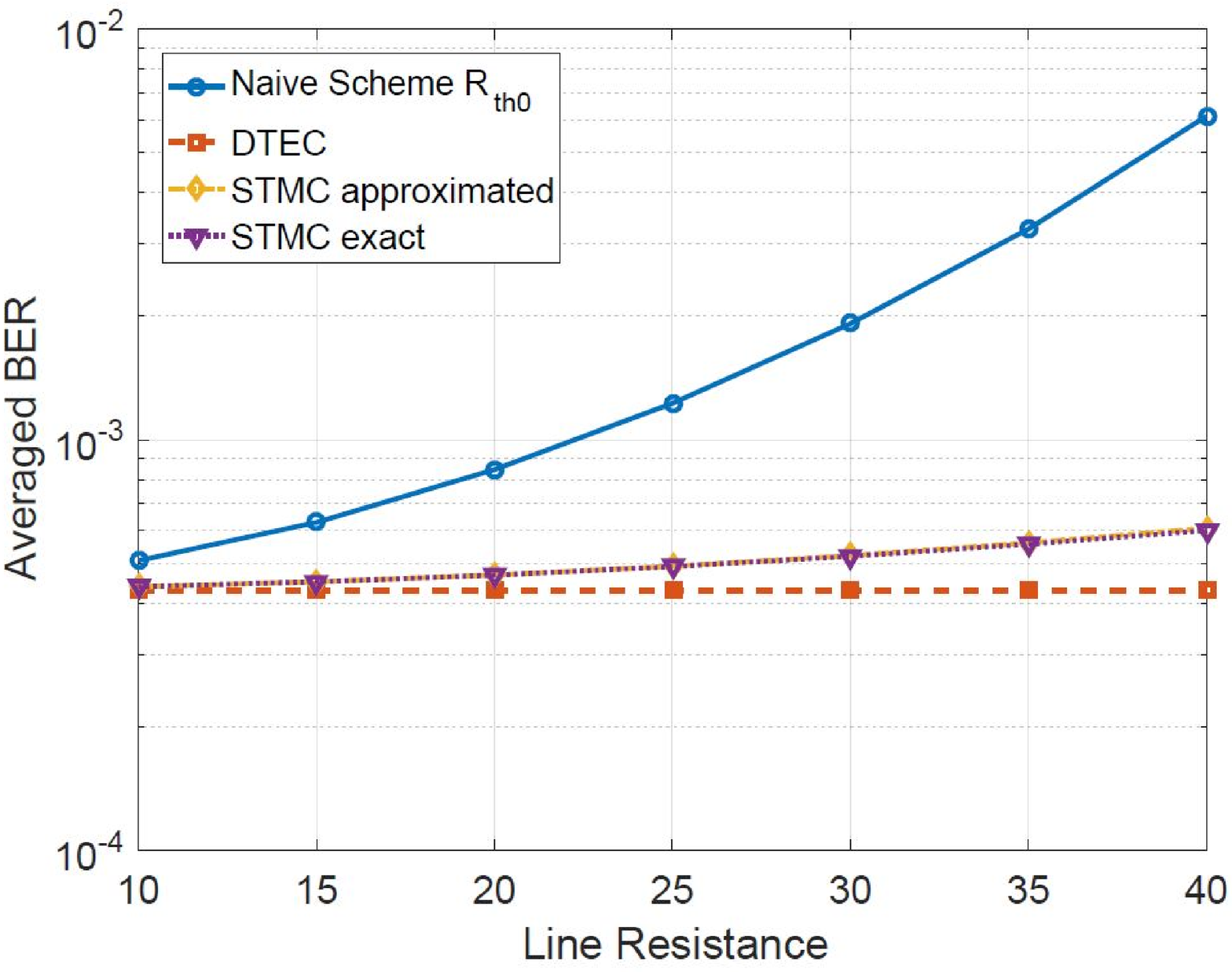}
	\caption{Averaged read BER v.s. wire resistance using different thresholding scheme.}
	\label{fig:thr_BER_vs_r_wire}
	\end{subfigure}
	\caption{Optimal Threshold Simulation Results}
	\vspace{-1em}
\end{figure}

From both Fig. \ref{fig:thr_BER_vs_n} and Fig. \ref{fig:thr_BER_vs_r_wire}, we observe that both DTEC and STMC schemes reduce the read BER significantly compared with the naive thresholding scheme. The DTEC scheme compensates for each cell based on its location, and, as a result, the averaged BER is independent of the array size and wire resistance. As expected, the averaged BERs using the exact solution of the STMC scheme is smaller than the averaged BERs using the approximated solution. However, the improvement is incremental and can not be identified on the plots. This shows that equation (\ref{equ:log_equ_sol_1}) approximates the solution of equation (\ref{equ:log_equ}) very well.
\section{Conclusion and Future Works}
In this paper, we proposed the write and read channel models for the 1S1R crossbar resistive memory while considering the nondeterministic nature of the memory device. We studied the optimal read threshold which reduces the row bit-error rate efficiently. Future research is in the direction of leveraging the channel information to improve memory reliability. This includes using error-correction codes with different correction capability for different rows and incorporating the location dependent channel information in LDPC decoding.
\section*{Acknowledgment}
Research supported in part by a grant from UC MEXUS and an NSF-BSF grant no.1718389.

% Generated by IEEEtran.bst, version: 1.14 (2015/08/26)

\end{document}